\documentclass{birkjour}
\usepackage{amssymb,amsthm,slashed,mathabx,bbold}

\usepackage[T1]{fontenc}
\usepackage[cp1250]{inputenc}
\usepackage{cite}
\usepackage{bm}
\usepackage{enumitem}

\usepackage[showonlyrefs]{mathtools}
\mathtoolsset{showonlyrefs=true}
\usepackage{fixltx2e}
\MakeRobust{\eqref}

\usepackage{hyperref}

\newcommand{\Hc}{\mathcal{H}}
\newcommand{\Dc}{\mathcal{D}}
\newcommand{\mR}{\mathbb{R}}

\newcommand{\inc}{\mathrm{in}}
\newcommand{\out}{\mathrm{out}}
\newcommand{\ret}{\mathrm{ret}}
\newcommand{\adv}{\mathrm{adv}}

\newcommand{\wt}{\widetilde}

\newcommand{\al}{\alpha}
\newcommand{\ga}{\gamma}

\newcommand{\la}{\lambda}
\newcommand{\si}{\sigma}
\newcommand{\w}{\omega}
\newcommand{\W}{\Omega}

\newcommand{\vep}{\varepsilon}
\newcommand{\dsp}{\displaystyle}

\newcommand{\ov}{\overline}
\newcommand{\p}{\partial}
\newcommand{\pav}{\bm{\partial}}
\newcommand{\piv}{\bm{\pi}}

\newcommand{\con}{\mathrm{const}}

\newcommand{\tm}{\langle t\rangle}
\newcommand{\av}{\mathbf{a}}

\newcommand{\ev}{\mathbf{e}}
\newcommand{\bv}{\mathbf{b}}

\newcommand{\eti}{\tilde{\mathbf{e}}}

\newcommand{\xti}{\tilde{\mathbf{x}}}
\newcommand{\piti}{\tilde{\bm{\pi}}}
\newcommand{\hti}{\tilde{h}}
\newcommand{\Ev}{\mathbf{E}}
\newcommand{\Bv}{\mathbf{B}}
\newcommand{\Cv}{\mathbf{C}}

\newcommand{\lv}{\mathbf{l}}
\newcommand{\dav}{\dot{\mathbf{a}}}

\newcommand{\x}{\mathbf{x}}
\newcommand{\qv}{\mathbf{q}}
\newcommand{\xm}{\langle |\mathbf{x}|\rangle}
\newcommand{\y}{\mathbf{y}}
\newcommand{\pv}{\mathbf{p}}
\newcommand{\Av}{\mathbf{A}}

\newcommand{\Piv}{\mathbf{\Pi}}

\DeclareMathOperator{\Rp}{Re}

\DeclareMathOperator*{\slim}{s-lim}

\newtheorem{thm}{Theorem}
\newtheorem{pr}[thm]{Proposition}
\newtheorem{lem}[thm]{Lemma}

\setlist[itemize]{leftmargin=2.5em}

\begin{document}

\title[Electromagnetic gauge choice for scattering of~Schr\"{o}dinger particle]%
{Electromagnetic gauge choice\\ for scattering of~Schr\"{o}dinger particle}

\author{Andrzej Herdegen}

\address{Institute of Physics\\
    Jagiellonian University\\
    ul.\,S.\,{\L}ojasiewicza 11\\
    30-348 Krak\'{o}w\\
    Poland}

\email{herdegen@th.if.uj.edu.pl}
{\it}\date{}



\date{}

\begin{abstract}

We consider a Schr\"{o}dinger particle placed in an external
electromagnetic field of the form typical for scattering settings in the
field theory: $F=F^\ret+F^\inc=F^\adv+F^\out$, where the current producing
$F^{\ret/\adv}$ has the past and future asymptotes homogeneous of degree
$-3$, and the free fields $F^{\inc/\out}$ are radiation fields produced by
currents with similar asymp\-totic behavior. We show that with appropriate
choice of electromagnetic gauge the particle has `in' and `out' states
reached with no further modification of the asymptotic dynamics. We use
a~special quantum-mechanical evolution `picture' in which the free
evolution operator has well-defined limits for $t\to\pm\infty$, thus the
scattering wave operators do not need the free evolution counteraction. The
existence of wave operators in this setting is established, but the proof
of asymptotic completeness is not complete: more precise characterization
of the asymptotic behavior of the particle for $|\x|=|t|$ would be needed.

\end{abstract}

\maketitle

\section{Introduction}\label{int}

Infrared problems are typical for theories with long-range interactions, and
extend over wide range of physical settings. They are particularly persistent
in the relativistic quantum theory--quantum field theory--where their nature
is not only technical, but also conceptual. The standard procedure adopted in
mathematically oriented formulations of the quantum electrodynamics (and
other theories with long range interaction) is to use local potentials (of
Gupta-Bleuler type) with adiabatically truncated interaction. One argues that
this setting is sufficient to construct (perturbatively) the local algebra of
observables of the theory. However, the removal of the cut-off is a singular
operation, which has the consequence that the states of the desired theory
cannot be those in which the truncated theory is constructed. This is
well-known and generally accepted. But one can also ask, whether the
truncation of interaction does not remove some structure from the theory in
an irreversible way; the algebraic equivalence of local algebras, for all
cutoff functions equal to one in the considered region, is based on an
interpolating relation, which becomes singular in the limit of the function
tending to unity on the whole spacetime.\footnote{For recent results on
adiabatic limit in quantum field theory, as well as a review of literature,
see \cite{duch18}.}

Considerations similar to those described above has motivated the pre\-sent
author, many years ago, to attempts to include, from the beginning, the
long-range degrees of freedom of the theory in the description. These
attempts went in two main directions: (i) construction of a nonlocal
electromagnetic potential in which scattering of a Dirac particle is infrared
nonsingular (see \cite{her95} for the analysis at the level of classical
fields); and (ii) extension of the algebra of the free quantum
electromagnetic field which includes long-range degrees of freedom, and which
is thought of as an asymptotic algebra, potentially starting point for
perturbation calculus (see \cite{her98} and a recent synopsis \cite{her17}).

The present article is a further test of the idea mentioned in (i) above. In
article \cite{her95} I have considered the scattering of the classical Dirac
field in an external electromagnetic field of the type supposed to be present
in full interacting theory. It was shown that if an appropriate (nonlocal, in
general) electromagnetic gauge is chosen, then the Dirac field has a
well-defined asymptote in remote future (and past) inside the lightcone. This
asymptote is reached without further corrections of asymptotic dynamics, as
usually employed in long-range scattering. Moreover, the scattered outgoing
Dirac field is constructed with the use of this asymptote. The hope behind
this analysis is that a similar construction in quantum case could similarly
relieve some of the infrared problems.

The article mentioned above lacks the discussion of the asymptotic behavior
of the Dirac field on the whole hyperplanes of constant time, for this time
tending to infinity. In the present paper we analyze this question in the
case of nonrelativistic Schr\"odinger particle. We show the existence of
asymptotic velocity operators and of isometric wave operators. The null
asymp\-totic behavior of radiation fields (and their potentials), together
with all their derivatives, is only of (1/time) type, which breaks the usual
assumptions imposed on time-decaying potentials considered in Schr\"odinger
scattering (see \cite{der97}, \cite{yaf10}). In the present setting this
behavior has prevented the proof of asymptotic completeness. Whether this can
be overcome is an open question.\footnote{Long-range scattering in the
position representation has been discussed before, see e.g.\ \cite{der97a}
and \cite{yaf10}, but not for the context considered here, and with no
relation to the choice of electromagnetic gauge.}

The choice of gauge found appropriate for the problem of \cite{her95} was
such that $x\cdot A(x)$ (Minkowski product of the position vector with the
potential) vanishes sufficiently fast in remote future (and past) inside the
lightcone \mbox{($x^2>0$)}. Here we shall construct an appropriate gauge in
the whole spacetime, and will see that its behavior inside the lightcone is
of the same type as before.

The method used for our analysis is the transformation of the time evolution
from the Schr\"o\-din\-ger picture to a new `picture', which we describe in
Section \ref{harmosc}, and in which the state vector of a particle tends to
its spacetime asymptotic forms in asymptotic times. As it turns out, the
simplest choice of such transformation is the well-known Niederer
transformation to the harmonic oscillator system \cite{nie73}. In Section
\ref{oscelmg} this oscillator is placed in electromagnetic field, and in
Section \ref{osctosch} transformed back to the Schr\"odinger picture. In
order to construct the time-dependent hamiltonian, and the evolution it
generates, we follow an article by Yajima \cite{ya11} in its use of the Kato
theorem.\footnote{There is also some similarity in the use of two time
evolutions related by unitary transformations. However, the transformations
are quite different: it is a pure gauge transformation in \cite{ya11}, while
here it is the transformation to the new `picture' mentioned above. Also, the
aims of these operations are quite different: Yajima's goal is to include
electromagnetic fields as singular as possible (both locally, and in
infinity), while here we are interested in scattering theory.} In Section
\ref{fieldsgauges} we discuss in detail the relation between potentials and
fields in the two pictures, and define a gauge appropriate for the
description of scattering along the lines described above. Section
\ref{scatteringosc} contains theorems on scattering in oscillator picture.
Reformulation of these results in natural spacetime terms and some final
remarks are contained in Section~\ref{scatspti}. Appendix~\ref{inequalities}
discusses some relations between domains of operators needed in the main
text. In Appendix \ref{aselm} we discuss the scope of electromagnetic fields
admitted in the article, and analyse, for completeness, their spacetime
estimates. Appendix~\ref{difk} clarifies a particular differentiation
employed in Section~\ref{scatteringosc}.

\section{Harmonic oscillator picture}\label{harmosc}

Let $\Hc=L^2(\mR^3)$ and denote $H_0=-\tfrac{1}{2}\Delta$, the free particle
hamiltonian with the corresponding evolution operator $U_0(t,s)$.\footnote{We
set $\hbar=1$, $c=1$ and use dimensionless rescaled quantities; to recover
physical quantities one should substitute $(t,\x)\mapsto(mt,m\x)$, with $m$
the mass of the particle. Also, the electromagnetic potentials to appear
later should be multiplied by $q/m$, with $q$ the charge of the particle.}
The free particle wave packet $\psi(t)=U_0(t,0)\psi$ is given by the Fourier
representation
\begin{equation}\label{free}
  \psi(t,\x)=(2\pi)^{-3/2}\int \exp\big[-\tfrac{i}{2}t|\mathbf{p}|^2+i\mathbf{p}\cdot\x\big]\hat{\psi}(\mathbf{p})\,d^3p\,,
\end{equation}
and for $\hat{\psi}$ regular enough, by the use of stationary phase method,
has the following asymptotic forms for $|t|\to\infty$:
\begin{equation}\label{stas}
  \psi(t,|t|\y)\sim |t|^{-3/2}e^{\mp i3\pi/4}\exp\big[\tfrac{i}{2}t|\y|^2\big]\hat{\psi}(\pm \y)\quad\text{for}\quad t\to\pm\infty\,.
\end{equation}

We would like to find a new evolution `picture', in which the state vector of
the particle has well defined limits for $t\to\pm\infty$. Using the
asymptotic forms \eqref{stas} as a guideline, we define for $t\in\mR$ the
following transformation in~$\Hc$:
\begin{equation}\label{N}
  [N(t)\chi](\x)=\tm^{-3/2}\exp\Big[\frac{it}{2}\Big(\frac{|\x|}{\tm}\Big)^2\Big]\chi\Big(\frac{\x}{\tm}\Big)\,,
\end{equation}
where $\tm>0$ is a $C^1$-function of $t$, such that $\tm/|t|\to1$ for
$|t|\to\infty$, whose exact form is to be determined. This is a unitary
transformation and the conjugate transformation is
\begin{equation}\label{N*}
  [N(t)^*\psi](\y)=\tm^{3/2}\exp\big[-\tfrac{i}{2}t|\mathbf{y}|^2\big]\psi(\tm\y)\,.
\end{equation}
We use the family of operators $N(t)$ to transform the Schr\"{o}dinger state
vectors and observables respectively by
\begin{equation}\label{hop}
  \psi_S(t)\mapsto \psi_N(t)=N(t)^*\psi_S(t)\,,\quad A_S(t)\mapsto A_N(t)=N(t)^*A_S(t)N(t)\,.
\end{equation}
The evolution operator for $\psi_N(t)$:
\begin{equation}\label{U0N}
  U_{0N}(t,s)=N(t)^*U_0(t,s)N(s)
\end{equation}
is strongly continuous, and for $\chi\in C^\infty_0(\mR^3)$ the vector
$U_{0N}(t,s)\chi$ is strongly differentiable in $t$ and in $s$, and a
straightforward calculation shows that
\begin{equation}\label{dU0N}
  i \p_t U_{0N}(t,s)\chi=H_{0N}(t)U_{0N}(t,s)\chi\,,
\end{equation}
where  (with $\tm'=d\tm/dt$)
\begin{multline}\label{H0N}
  H_{0N}(t)=N(t)^*H_0N(t)-iN(t)^*\p_tN(t)\\
  =\frac{1}{2\tm^2}\Big([\pv^2+(\tm^2-t^2)\x^2]+(t-\tm\tm')[\x\cdot\pv+\pv\cdot\x+2t\x^2]\Big)\,;
\end{multline}
here and in what follows $\pv=-i\pav$. A large class of functions $\tm$ for
which $t-\tm\tm'$ tends to zero and $\tm^2-t^2$ tends to a constant
sufficiently fast would fulfil our demands formulated after \eqref{stas}.
However, the simplest choice is to demand $t-\tm\tm'=0$, in which case
$\tm^2=t^2+\kappa^2$, where $\kappa^2$ is a~constant. The choice of this
constant is physically irrelevant, and as our variables are dimensionless,
the simplest choice is $\kappa=1$, so that $N(0)=\mathbb{1}$. The new picture
hamiltonian is now
\begin{equation}
 H_{0N}(t)=\frac{1}{2\tm^2}[\pv^2+\x^2]\,,\quad \tm=\sqrt{t^2+1}\,,
\end{equation}
and the choice of $\tm$ holds in the rest of the article. If we now define
\begin{equation}\label{u0}
  u_0(\tau,\si)=U_{0N}(\tan\tau,\tan\si)\,,
\end{equation}
then
\begin{equation}\label{u0h}
  u_0(\tau,\si)=\exp[-i(\tau-\si)h_0]\,,\qquad h_0=\tfrac{1}{2}[\pv^2+\x^2]\,.
\end{equation}
The evolution of the free particle state over $t\in(-\infty,+\infty)$
corresponds now to the evolution over $\tau\in(-\pi/2,+\pi/2)$, thus over
half the period of the harmonic oscillator. We shall call this new	
description of evolution \emph{the harmonic  oscillator picture}.  The
scattering change of state is described in this picture by the operator
\begin{equation}\label{s}
  \slim_{\tau\nearrow\pi/2}u_0(+\tau,-\tau)=u_0(+\pi/2,-\pi/2)=\exp[-i\pi h_0]=iP\,,
\end{equation}
where $P$ is the parity operator $[P\chi](\x)=\chi(-\x)$. This formula is in
agreement with the asymptotic forms \eqref{stas}.

The transformation $N(t)$ leading from a~free particle to an harmonic
oscillator has been discovered much earlier by Niederer \cite{nie73}. It is
interesting to note, that his original derivation was a result of
group-theoretical considerations, with no relation to scattering theory.
Although the Niederer transformation has found numerous applications in
literature, to our knowledge it has not appeared in the scattering of a
Schr\"odinger particle context before. Which, if true, would be quite
surprising, if one notes its striking similarity to the asymptotic form
\eqref{stas}.

We shall later see that the above relation extends to a system placed in
external electromagnetic field: a charged particle dynamics becomes a charged
oscillator in the new picture (with electromagnetic potentials appropriately
transformed). For technical reasons, we find it convenient to start our
discussion in the oscillator picture, and only then transform into
Schr\"{o}dinger picture.

\section{Harmonic oscillator in electromagnetic field}\label{oscelmg}

We consider in $\Hc=L^2(\mR^3)$  the hamiltonian
\begin{equation}
 h=\tfrac{1}{2}(\piv^2+\x^2)+v\,,\quad \piv=\pv-\av\,,
\end{equation}
where $\av(\tau,\x)$ and $v(\tau,\x)$ are electromagnetic potentials. This
hamiltonian may be given the precise meaning for a wide class of potentials
(see \cite{ya11}), but for our purposes it is sufficient to consider the
following setting.

First note, that for vanishing potentials the corresponding `unperturbed'
harmonic oscillator hamiltonian  $h_0=\tfrac{1}{2}(\pv^2+\x^2)$ is determined
by the closed quadratic form\footnote{Conditions like $\pv\psi\in\Hc$ and
$\x\psi\in\Hc$ below, and similar other to appear further with vector
quantities on the lhs, are meant to hold component-wise.}
\begin{gather}
 q_0(\varphi,\psi)=\tfrac{1}{2}(\pv\varphi,\pv\psi)+\tfrac{1}{2}(\x\varphi,\x\psi)\,,\\
 \Dc(q_0)=\{\psi\in\Hc|\,\pv\psi\in \Hc, \x\psi\in \Hc\}=\Dc(\pv)\cap\Dc(\x)=\Dc(h_0^{1/2})\,\label{formq0}
\end{gather}
and the following relations hold
\begin{gather}
 \Dc(h_0)=\Dc(\pv^2)\cap\Dc(\x^2)\subseteq\Dc(\x\cdot\pv+\pv\cdot\x)\label{domh0}
\end{gather}
(see Appendix \ref{inequalities}).\footnote{It is quite surprising, that the
characterization of the harmonic oscillator domain given by the equality in
\eqref{domh0} is rather hard to find in standard textbook discussions of the
harmonic oscillator.} Moreover, $C^\infty_0(\mR^3)$ is a form-core and a core
for $h_0$ (see Thm.\,X.28 in \cite{rs79II}).

We also note an elementary inequality which will be used below. Let
$\psi(\tau)\in\Hc$ depend differentiably on a real parameter~$\tau$. Then
\begin {equation}\label{ineqdiff}
 |\p_\tau\|\psi(\tau)\||\leq\|\p_\tau\psi(\tau)\|\,.
\end{equation}
To show this, note that
$\p_\tau\|\psi(\tau)\|^2=2\Rp(\psi(\tau),\p_\tau\psi(\tau))$ and use the
Schwarz inequality for the rhs.

\begin{pr}\label{defh}
Let $\av$ and $v$ be in $L^\infty(\mR^3)$ (for each fixed time) and define
the quadratic form
\begin{equation}
 q(\varphi,\psi)=\tfrac{1}{2}(\piv\varphi,\piv\psi)+\tfrac{1}{2}(\x\varphi,\x\psi)
 +(\varphi,v\psi)\equiv q_0(\varphi,\psi)+\rho(\varphi,\psi)\,,
\end{equation}
with the form-domain $\Dc(q)=\Dc(q_0)$. Then $q$ is a closed form
corresponding to the unique self-adjoint operator $h$, for which
$C^\infty_0(\mR^3)$ is a form core. If~\mbox{$\al>\|v\|_\infty$}, then
$h+\al\mathbb{1}$ is positive, and
\mbox{$\Dc((h+\al\mathbb{1})^{1/2})=\Dc(q_0)$}.

If, in addition, $(\pav\cdot\av)$ is also in $L^\infty(\mR^3)$, then
\mbox{$\Dc(h)=\Dc(h_0)$}, and $C^\infty_0(\mR^3)$ is a domain of essential
self-adjointness of $h$.
\end{pr}
\begin{proof}
For $\varphi\in \Dc(q_0)$ one has
\begin{equation}\label{}
  |\rho(\varphi,\varphi)|\leq \|\av\varphi\|\|\pv\varphi\|
  +(\varphi,[\tfrac{1}{2}\av^2+|v|]\varphi)
  \leq \tfrac{1}{2}q_0(\varphi,\varphi)+c_1\|\varphi\|^2\,.
\end{equation}
Therefore, the first part of the theorem follows by the KLMN theorem (see
e.g.\ \cite{rs79II}, Thm.\,X.17). If, in addition, $(\pav\cdot\av)\in
L^\infty(\mR^3)$, then for $\varphi\in \Dc(q_0)$ and $\psi\in\Dc(h_0)$ we
have
 $\rho(\varphi,\psi)=(\varphi,r\psi)$ with
\begin{equation}\label{r}
 r=h-h_0=-\av\cdot\pv+\tfrac{i}{2}(\pav\cdot\av)+\tfrac{1}{2}\av^2+v\,,
\end{equation}
 so
\begin{equation}\label{estR}
 \|r\psi\| \leq (\|\av\|_\infty\|\psi\|^{1/2})(\|\pv\psi\|\|\psi\|^{-1/2})+c_2\|\psi\|
 \leq \tfrac{1}{2}\|h_0\psi\|+c_3\|\psi\|\,.
\end{equation}
Thus the second part of the thesis follows by the Kato-Rellich theorem.
\end{proof}

The existence of the corresponding evolution operators is assured under the
conditions of the following theorem. Here and in the rest of the article the
overdot denotes differentiation with respect to time in the oscillator
picture.
\begin{thm}\label{evolution}
For an interval $I\subset\mR$ let the functions $\av$, $\pav\cdot\av$,
$\dav$, $\pav\cdot\dav$, $v$ and $\dot{v}$ be in $L^\infty(I\times\mR^3)$.
Then for $\tau,\,\si\in I$:
\begin{itemize}
 \item[(i)] all $h(\tau)$ satisfy Proposition \ref{defh};
 \item [(ii)] there exists the unique unitary propagator $u(\tau,\si)$ for
     the family $h(\tau)$, strongly continuous in $(\tau,\si)$, with the
     following properties:
     \begin{itemize}
     \item[(a)] $u(\tau,\si)\Dc(h_0)=\Dc(h_0)$;
     \item[(b)] for $\psi\in\Dc(h_0)$ the map $(\tau,\si)\mapsto
         u(\tau,\si)\psi$ is of class $C^1$ in the strong sense and
     \begin{align}
     i\p_\tau u(\tau,\si)\psi&=h(\tau)u(\tau,\si)\psi\,,\label{hu}\\
     i\p_\si u(\tau,\si)\psi&=-u(\tau,\si)h(\si)\psi\,.\label{uh}
     \end{align}
     \end{itemize}
 \end{itemize}
\end{thm}
\begin{proof}
 For $\psi\in \Dc(h_0)$ we have
\begin{equation}
 \dot{h}\psi=\dot{r}\psi=[-\dav\cdot\pv+\tfrac{i}{2}(\pav\cdot\dav)+\dav\cdot\av+\dot{v}]\psi\,,
\end{equation}
see \eqref{r}, so estimating as in \eqref{estR} we obtain
 $\|\dot{h}\psi\|\leq\tfrac{1}{2}\|h_0\psi\|+c_4\|\psi\|$ and
\begin{equation}\label{grcont}
 \|[h(\tau)-h(\si)]\psi\|\leq|\tau-\si|(\tfrac{1}{2}\|h_0\psi\|+c_4\|\psi\|)\,.
\end{equation}
Next, using \eqref{estR} we find
\begin{equation}
 \|h_0\psi\|=2\|[h-r]\psi\|-\|h_0\psi\|\leq2(\|h\psi\|+c_3\|\psi\|)\,,
\end{equation}
  so
\begin{equation}
 \|\dot{h}\psi\|\leq \|h\psi\|+c_5\|\psi\|\,.
\end{equation}
  Let us denote
\begin{equation}
 \|\psi\|_\tau\equiv[\|h(\tau)\psi\|^2+\|\psi\|^2]^{1/2}
\end{equation}
 --the graph norm of
$h(\tau)$. The use of \eqref{ineqdiff} now gives
\[
 |\p_\tau\|\psi\|_\tau^2|
 \leq 2\|h\psi\|\|\dot{h}\psi\|\leq2 c\|\psi\|_\tau^2
\]
for some $c>0$. This implies $\|\psi\|_\tau\leq e^{c|\tau-\si|}\|\psi\|_\si$,
and the way is paved for the application of the Kato theorem in the form
given by Theorem 3.2 in \cite{ya11}. Denote $\mathcal{Y}=\Dc(h_0)$ equipped
with the graph norm of $h_0$, and \mbox{$\mathcal{X}=\Hc$}.
Then~\eqref{grcont} says that $h(\tau):\mathcal{Y}\mapsto\mathcal{X}$ is
norm-continuous. Moreover, let $\mathcal{Y}_\tau$ be $\Dc(h_0)$ equipped with
the graph norm  $\|.\|_\tau$, and identify $A(\tau)=h(\tau)$. Then all
conditions of this theorem are satisfied and one obtains $u(\tau,\si)$ with
the stated properties.
\end{proof}

\section{From oscillator to Schr\"{o}dinger picture}\label{osctosch}

The way back to the Schr\"{o}dinger picture is achieved as follows.
\begin{thm}\label{OS}
 Let $\av$ and $v$ satisfy the assumptions of Theorem \ref{evolution} on each
 compact interval $I\subset(-\pi/2,+\pi/2)$. Denote
\begin{gather}
\begin{aligned}\label{AaVv}
 \Av(t,\x)&=\tm^{-1}\av(\arctan t,\tm^{-1}\x)\,,\\[1ex]
 V(t,\x)&=\tm^{-2}[v(\arctan t, \tm^{-1}\x)+\tm^{-1}t\x\cdot\av(\arctan t,\tm^{-1}\x)]\,,
 \end{aligned}\\[1ex]
 U(t,s)=N(t)u(\arctan t, \arctan s)N(s)^*\,.\label{Uu}
\end{gather}
Then $\mR^2\ni(t,s)\mapsto U(t,s)$ is a strongly continuous family of unitary
evolution operators, with the properties:
\begin{itemize}
\item[(i)] $U(t,s)\Dc(h_0)=\Dc(h_0)$;
\item[(ii)] for $\psi\in\Dc(h_0)$ the map $(t,s)\mapsto U(t,s)\psi$ is of
    class $C^1$ in the strong sense and
    \begin{equation}\label{evolU}
    i\p_tU(t,s)\psi=H(t)U(t,s)\psi\,,\quad
    i\p_sU(t,s)\psi=-U(t,s)H(s)\psi\,,
    \end{equation}
where
\begin{equation}
 H(t)=\tfrac{1}{2}\Piv(t)^2+V(t)\,,\quad \Piv=\pv-\Av\,.
\end{equation}
Operators $H(t)$ are essentially self-adjoint on $C^\infty_0(\mR^3)$.
\end{itemize}
The inverse relations of potentials are
\begin{align}\label{}
  \av(\tau,\x)&=\tm \Av(t,\tm\x)|_{t=\tan\tau}\,,\\
  v(\tau,\x)&=\tm^2\big[V(t,\tm\x)-(t/\tm)\x\cdot\Av(t,\tm\x)\big]|_{t=\tan\tau}\,.
\end{align}
\end{thm}
\begin{proof}
Taking into account relations \eqref{domh0} one finds
$N(s)^*\Dc(h_0)=\Dc(h_0)$ and $N(t)\Dc(h_0)=\Dc(h_0)$, which proves (i). For
$\psi\in\Dc(h_0)$ and $F$--the operator of multiplication by the function
$F(\x)$, we have
\begin{gather}
 [N(t)FN(t)^*\psi](\x)=F(\x/\tm)\psi(\x)\,,\label{NFN}\\
 [N(t)\pv N(t)^*\psi](\x)=\big[\tm\pv-(t/\tm)\x\big]\psi(\x)\,,
\end{gather}
so a straightforward calculation gives
\begin{multline}
 [N(t)h(\arctan t)N(t)^*\psi](\x)\\= \tm^2\big[\tfrac{1}{2}(\pv-A(t,\x))^2+V(t,\x)\big]\psi(\x)
 +\tfrac{1}{2}[\x^2-t(\x\cdot\pv+\pv\cdot\x)]\psi(\x)\,,
\end{multline}
with $\Av$ and $V$ given in the thesis. On the other hand
\begin{equation}
 -i[N(t)\p_tN(t)^*\psi](\x)=\tfrac{1}{2}\tm^{-2}[-\x^2+t(\x\cdot\pv+\pv\cdot\x)]\psi(\x)\,.
\end{equation}
Thus equations \eqref{evolU} are satisfied with
\begin{equation}
 H(t)=\tm^{-2}N(t)h(\arctan t)N(t) -i N(t)\p_tN(t)=\tfrac{1}{2}\Piv(t)^2+V(t)\,.
\end{equation}
For each $t$ the norm $\|\Av(t,.)\|_\infty$ is finite, and if
$\|V(t,.)\|_\infty$ is finite as well, then the essential self-adjointness
follows easily in standard way (as in the proof of Proposition \ref{defh}).
In general, there is only $\|\xm^{-1}V(t,.)\|_\infty<\infty$ (due to the
$\x\cdot\av$ term in $V$), but then the use of the Leinfelder-Simader theorem
(Thm.\,4 in \cite{ls81}) leads to the same conclusion.
\end{proof}

\section{Electromagnetic fields and gauges}\label{fieldsgauges}

We shall now discuss in detail the relation between $(V,\Av)$ and $(v,\av)$
defined in \eqref{AaVv}. We assume that all differentiations to appear may be
performed, but in the following theorem the assumptions of Theorem~\ref{OS}
are not needed.
\begin{pr}\label{figa}
 Let $(V,\Av)$ and $(v,\av)$ be related by \eqref{AaVv}. Denote the electric
 and magnetic fields of these potentials by $(\Ev,\Bv)$ and $(\ev,\bv)$,
 respectively.
 Then the following is satisfied:
 \begin{itemize}
 \item[(i)] The gauge transformation
  \begin{equation}
 \Av_\Lambda(t,\x)=\Av(t,\x)-\pav\Lambda(t,\x)\,,\quad
 V_\Lambda(t,\x)= V(t,\x)+\p_t\Lambda(t,\x)\,,
 \end{equation}
 is equivalent to the transformation
 \begin{equation}
 \av_\la(\tau,\x)=\av(\tau,\x)-\pav\la(\tau,\x)\,,\quad
 v_\la(\tau,\x)=v(\tau,\x)+\p_\tau\la(\tau,\x)
 \end{equation}
 where $\lambda(\tau,\x)=\Lambda(t,\tm \x)|_{t=\tan\tau}$.
 \item[(ii)] The electromagnetic fields are related by:
 \begin{align}
 \ev(\tau,\x)&=\Big[\tm^3\Ev(t,\tm\x)+t\tm^2\x\times\Bv(t,\tm\x)\Big]_{t=\tan\tau}\,,\\
 \bv(\tau,\x)&=\tm^2\Bv(t,\tm\x)|_{t=\tan\tau}\,.
 \end{align}
 \item[(iii)] For given $(v,\av)$ choose the gauge function as
 \begin{equation}
   \la_\mathrm{e}(\tau,\x)=-\int_0^\tau v(\rho,\x)d\rho+\la_\mathrm{e}(0,\x)\,,
 \end{equation}
 where $\la_\mathrm{e}(0,\x)$ constitutes the remaining freedom in the
 definition. Then:
 \begin{equation}\label{figae}
 v_\mathrm{e}(\tau,\x)=0\,,\qquad \av_\mathrm{e}(\tau,\x)
 =\av_\mathrm{e}(0,\x)-\int_0^\tau\ev(\rho,\x)d\rho\,,
 \end{equation}
which we shall call an \emph{$\av$-gauge}. In these gauges the assumptions
of\linebreak Theorem~\ref{OS} are reduced to the following:
$\av_\mathrm{e}(0),\,\pav\cdot\av_\mathrm{e}(0)\in L^\infty(\mR^3)$ and
\mbox{$\ev,\,\pav\cdot\ev\in L^\infty(I\times\mR^3)$} for each compact
interval $I\subset (-\pi/2,+\pi/2)$.
\item[(iv)] In all $\av$-gauges the four-potential
    $A^a(x)=(V(t,\x),\Av(t,\x))$ satisfies
    \begin{equation}\label{xA}
    \hat{x}\cdot  A(x)=0\,,\qquad \hat{x}=x+((x^0)^{-1},\mathbf{0})
    \end{equation}
(there is no singularity at $x^0=0$).
 \end{itemize}
\end{pr}
\begin{proof}
All properties follow by simple calculations, which we leave to the reader.
\end{proof}

Preparing to discuss scattering, we need to formulate the asymptotic behavior
of potentials and fields: for upper case fields in the limit $|t|\to \infty$,
and the corresponding behavior of lower case fields in the limit
$|\tau|\to\pi/2$. It~will be convenient to denote
\begin{equation}\label{C}
 \Cv(t,\x)=t\Ev(t,\x)+\x\times\Bv(t,\x)\,,
\end{equation}
so then
\begin{equation}\label{e}
 \ev(\tau,\x)=\Big[t\tm\Cv(t,\tm\x)+\tm\Ev(t,\tm\x)\Big]_{t=\tan\tau}\,.
\end{equation}

In what follows we shall always use $\av$-gauges, so we only need to consider
the  fields $\Ev$, $\Cv$ and $\ev$. In Appendix \ref{aselm} we summarize the
asymptotic properties of electromagnetic fields in scattering settings in
Minkowski space. For $\Ev$ and $\Cv$ it will be sufficient to note
\begin{equation}\label{EC}
 |\Ev(t,\x)|\leq\con\,,\quad |\Cv(t,\x)|\leq\frac{\con}{1+|t|+|\x|}\,,
\end{equation}
while for divergence of $\Cv$ we shall need the bound
\begin{equation}\label{divC}
 |\pav\cdot\Cv(t,\x)|\leq\frac{\con}{(1+|t|+|\x|)(1+||t|-|\x||)^\ga}
\end{equation}
for some $1\geq\ga>0$. A straightforward calculation with the use of
\eqref{e} shows that these relations imply the following lower case fields
bounds:
\begin{equation}\label{ebounds}
 |\ev(\tau,\x)|\leq\frac{\con}{\cos\tau}\,,\quad
 |\pav\cdot\ev(\tau,\x)|\leq\frac{\con}{(\cos\tau)^{2-\ga}(\cos\tau+||\sin\tau|-|\x||)^\ga}\,,
\end{equation}
where in the last formula we have taken into account that
$\pav\cdot\Ev(t,\x)=4\pi\rho(t,\x)$ (the charge density), which gives a
contribution to $\pav\cdot\ev$ bounded by a constant, which is less
restricting than the term in \eqref{ebounds}.

While reading the next section, the reader is asked to note that the bounds
\eqref{ebounds} are sufficient to satisfy the assumptions of Theorems
\ref{aspos} and \ref{wave}, but for Theorem \ref{ascompl} the second bound in
\eqref{ebounds} would have to be replaced by
\begin{equation}\label{eboundstr}
 |\pav\cdot\ev(\tau,\x)|\leq\frac{\con}{(\cos\tau)^{2-\ga}}\,.
\end{equation}

\section{Scattering in oscillator picture}\label{scatteringosc}

Turning to the discussion of scattering, we start again with the oscillator
picture. With increasingly restrictive assumptions we shall prove:
\begin{itemize}
\item[(i)]~existence of asymptotic position,
\item[(ii)]~existence of wave operators, and
\item[(iii)]~their unitarity
    (asymptotic completeness).
\end{itemize}
From now on we shall use only $\av$-gauges defined by \eqref{figae}, and we
omit the subscript~$\mathrm{e}$.

For any operator $k=k(\tau)$ with $\Dc(h_0)\subseteq\Dc(k(\tau))$ we denote
\begin{equation}\label{tik}
 \tilde{k}=\tilde{k}(\tau)=u(0,\tau)k(\tau)u(\tau,0)\,,\quad \text{so}\quad
\Dc(h_0)\subseteq\Dc(\tilde{k}(\tau))\,.
\end{equation}
We assume that the assumptions of Theorem \ref{OS}, as reformulated in
Proposition \ref{figa} (iii), are satisfied. Then for
\mbox{$\psi\in\Dc(h_0)$}, following the remarks in Appendix \ref{difk}, one
finds
\begin{align}
 \dot{\hti}\psi&=\tfrac{1}{2}[\eti\cdot\piti+\piti\cdot\eti]\psi\,,\label{hti}\\
 \dot{\xti}\psi&=\piti\psi\,,\label{xti}\\
 \p_\tau f(\xti)\psi&=(\pav f)(\xti)\cdot\piti\psi-\tfrac{i}{2}(\Delta f)(\xti)\psi\,,\label{fxti}
\end{align}
where $f$ is any smooth bounded function with bounded derivatives. Using
\eqref{hti} we obtain
\begin{equation}
 |\p_\tau\|\hti^{1/2}\psi\|^2|=|(\psi,\dot{\hti}\psi)|
 \leq\|\eti\psi\|\|\piti\psi\|
 \leq\sqrt{2}\|\eti\psi\|\|\hti^{1/2}\psi\|\,,
\end{equation}
so
\begin{equation}
 |\|\hti(\tau)^{1/2}\psi\|-\|h(0)^{1/2}\psi\||
 \leq\tfrac{1}{\sqrt{2}}\int_0^\tau\|\eti(\si)\psi\| d\si\,.
\end{equation}
It follows, in particular, that
\begin{equation}\label{pibound}
 \|\piti(\tau)\psi\|\leq \sqrt{2}\,\|h(0)^{1/2}\psi\|
 +\int_0^\tau\|\ev(\si,.)\|_\infty d\si\|\psi\|\,.
\end{equation}

\begin{thm}\label{aspos}
Let the electromagnetic fields and potentials satisfy the assumptions of
Theorem \ref{OS} in the form given in Proposition \ref{figa} (iii). If the
field $\ev$ satisfies
\begin{equation}\label{posas}
 \int_{-\pi/2}^{\pi/2}\big(\tfrac{\pi}{2}-|\tau|\big)\|\ev(\tau,.)\|_\infty\,d\tau<\infty \,,\\
\end{equation}
then there exist the asymptotic position operators $\xti_\pm$ defined by
\begin{equation}\label{pos}
 \exp{\big[i\qv\cdot \xti_\pm]}=\slim_{\tau\to\pm\pi/2}\exp{\big[i\qv\cdot\xti(\tau)\big]}
\end{equation}
for all numerical vectors $\qv$. The limit then holds for any Schwartz
function.
\end{thm}
\begin{proof}
Using equation \eqref{fxti} we obtain
\begin{equation}\label{exppi}
 \big\|\p_\tau\exp{\big[i\qv\cdot\xti(\tau)\big]}\psi\big\|
 =\|[\qv\cdot\piti(\tau)+\tfrac{1}{2}\qv^2]\psi\|
 \leq|\qv|\|\piti(\tau)\psi\|+\tfrac{1}{2}|\qv|^2\|\psi\|\,.
\end{equation}
Noting that the integration by parts gives
\begin{equation}\label{parts}
 \int_0^{\pi/2}\int_0^\tau\|\ev(\si,.)\|_\infty d\si d\tau
 \leq\int_0^{\pi/2}\big(\tfrac{\pi}{2}-\tau\big)\|\ev(\tau,.)\|_\infty d\tau\,,
\end{equation}
we observe that the rhs of \eqref{pibound} and \eqref{exppi} are integrable
over $[ 0,\pi/2]$. Therefore, the limit
$\dsp\lim_{\tau\to\pi/2}\exp[i\qv\cdot\xti(\tau)]\psi$ exists for each
$\psi\in\Dc(h_0)$, hence for all~$\psi$. The existence of a self-adjoint
operator $\xti_+$ satisfying the limiting relation \eqref{pos}
follows.\footnote{This is a simple consequence, which is left to the reader
as Problem VIII.23 in \cite{rs75I}.} The case of $\xti_-$ is similar.
\end{proof}

We now turn to the existence of wave operators.
\begin{thm}\label{wave}
Let the electromagnetic fields and potentials satisfy again the assumptions
given in Proposition \ref{figa} (iii). Suppose that
\begin{equation}\label{wave1}
 \int_{-\pi/2}^{\pi/2}\big(\tfrac{\pi}{2}-|\tau|\big)^2\|\ev(\tau,.)^2\|_\infty\,d\tau<\infty
\end{equation}
and for each closed set $K\subset \mR^3\setminus\{\x: \x^2=1\}$, with
notation $\|.\|_{K,\infty}\equiv\|.\|_{L^\infty(K)}$, there is
\begin{equation}\label{wave2}
 \int_{-\pi/2}^{\pi/2}\big(\tfrac{\pi}{2}-|\tau|\big)\|\pav\cdot\ev(\tau,.)\|_{K,\infty}\,d\tau<\infty\,.
\end{equation}
Then there exist isometric wave operators
\begin{equation}\label{waveop}
 u(0,\pm \pi/2)\equiv \slim_{\tau\to\pm \pi/2}u(0,\tau)\equiv \w_{\pm}\,,
\end{equation}
such that for each bounded, continuous function $f(\x)$ with support not
intersecting $\x^2=1$ there is
\begin{equation}\label{wavepos}
 f(\xti_\pm)=\w_\pm f(\x)\w_\pm^*\,.
\end{equation}
It follows that
\begin{equation}\label{isom}
 \w_\pm \w_\pm^*=E_\pm\equiv\mathbb{1}-\mathbb{1}_{\{1\}}(\xti_\pm^2)
\end{equation}
and
\begin{equation}\label{wave3}
 \slim_{\tau\to\pm\pi/2}u(\tau,0)E_\pm=\w_\pm^*\,.
\end{equation}
\end{thm}
\begin{proof}
Let $\psi$ be any smooth function with compact support $K$ outside
\mbox{$\x^2=1$}. If we show that $\|h(\tau)\psi\|$ is integrable over
$[-\pi/2,\pi/2]$, then by Eq.\eqref{uh} $u(0,\tau)\psi$ converges for
$\tau\to\pm\pi/2$. As the assumed class of functions is dense in $\Hc$, the
existence of $\w_\pm$ will be achieved. Now,
\begin{multline}\label{hbound}
 \|h(\tau)\psi\|\leq\|h_0\psi\|+\|\av(\tau)\cdot\pav\psi\|+\tfrac{1}{2}\|(\pav\cdot\av)(\tau)\psi\|
 +\tfrac{1}{2}\|\av(\tau)^2\psi\|\\
 \leq\|h_0\psi\|+\|\av(\tau,.)\|_\infty\|\pav\psi\|+
 \tfrac{1}{2}\big(\|\av(\tau,.)\|^2_\infty+\|\pav\cdot\av(\tau,.)\|_{K,\infty}\big)\|\psi\|\,.
\end{multline}
Using the Schwarz inequality and the gauge \eqref{figae} we have
\begin{equation}
 \frac{2}{\pi}\bigg[\int\limits_0^{\pi/2}\|\av(\tau,.)\|_\infty d\tau\bigg]^2
 \leq\int\limits_0^{\pi/2}\|\av(\tau,.)\|^2_\infty d\tau
 \leq \pi\|\av(0,.)\|_\infty^2+2\int\limits_0^{\pi/2}n(\tau)^2d\tau\,,
\end{equation}
where $n(\tau)=\int_0^\tau \|\ev(\si,.)\|_\infty d\si$. Integrating by parts
and then using the Schwarz inequality we obtain
\begin{multline}
 \int_0^{\pi/2}n(\tau)^2d\tau\leq2\int_0^{\pi/2}(\tfrac{\pi}{2}-\tau)n(\tau)\dot{n}(\tau)d\tau\\
 \leq 2\sqrt{\int_0^{\pi/2}n(\tau)^2d\tau}\sqrt{\int_0^{\pi/2}(\tfrac{\pi}{2}-\tau)^2\dot{n}(\tau)^2d\tau}
 \leq 4\int_0^{\pi/2}(\tfrac{\pi}{2}-\tau)^2\|\ev(\tau,.)\|_\infty^2d\tau\,,
\end{multline}
which shows that the second and third terms on the rhs of \eqref{hbound} are
integrable over $[0,\pi/2]$. Finally, integrating by parts as in
\eqref{parts} we find
\begin{equation}
 \int_0^{\pi/2}\|\pav\cdot\av(\tau,.)\|_{K,\infty}d\tau
 \leq \tfrac{\pi}{2}\|\pav\cdot\av(0,.)\|_{K,\infty}
 +\int_0^{\pi/2}(\tfrac{\pi}{2}-\tau)\|\pav\cdot\ev(\tau,.)\|_{K,\infty}d\tau\,,
\end{equation}
which completes the proof of the existence of $\w_+$ \eqref{waveop}. The case
of $\w_-$ is analogous.

The proof of the remaining claims is based on the following property, to be
proved below: for each smooth function $f(\x)$ with compact support $K$ not
intersecting $\x^2=1$, there is
\begin{equation}\label{smoothfw}
 \slim_{\tau\to\pm\pi/2}f(\x)u(\tau,0)=f(\x)\w_\pm^*\,.
\end{equation}
Once this is proved, we have
\begin{equation}
 f(\xti_\pm)=\slim_{\tau\to\pm\pi/2}u(0,\tau)f(\x)u(\tau,0)=\w_\pm f(\x)\w_\pm^*\,,
\end{equation}
where the first equality is the result of Theorem \ref{aspos}. Each bounded
continuous function with support outside $\x^2=1$, as well as the
characteristic function of the set $\mR^3\setminus\{\x|\x^2=1\}$, is a
point-wise limit of the above functions, so \eqref{wavepos} and \eqref{isom}
follow (for the latter note that $\mathbb{1}_{\{1\}}(\x^2)=0$). Finally,
there is
\begin{equation}
 \|u(\tau,0)E_\pm\psi-\w_\pm^*\psi\|^2
 =2\|E_\pm\psi\|^2-2\Rp(E_\pm\psi,u(0,\tau)\w_\pm^*\psi)\to 0
\end{equation}
for $\tau\to\pm\pi/2$, so \eqref{wave3} follows.

Turning to the proof of the remaining property \eqref{smoothfw} we use
\eqref{hu} to obtain for $\psi\in\Dc(h_0)$:
\begin{equation}\label{ufu}
 u_0(0,\tau)f(\x)u(\tau,0)\psi=f(\x)\psi
 -i\int_0^\tau u_0(0,\si)[fh(\si)-h_0f]u(\si,0)\psi d\si\,,
\end{equation}
where $u_0$ and $h_0$ are unperturbed operators, and without restricting
generality we assume $\|f\|_\infty\leq1$. Now,
\[
 fh-h_0f=(i\pav f-f\av)\cdot\piv+\tfrac{1}{2}f(i\pav\cdot\av-\av^2)
 +i(\pav f\cdot \av)+\tfrac{1}{2}(\Delta f)\,.
\]
The norm of the integrand in \eqref{ufu} may be thus bounded as
\begin{multline}
 \|[fh(\si)-h_0f]u(\si,0)\psi\|
 \leq(\|\pav f\|_\infty+\|\av(\si,.)\|_\infty)\|\piti(\si)\psi\|\\
 +\tfrac{1}{2}\Big[\|\pav\cdot\av(\si,.)\|_{K,\infty}
 +\|\av(\si,.)\|_\infty^2
 +2\|\pav f\|_\infty\|\av(\si,.)\|_\infty+\|\Delta f\|_\infty\Big]\|\psi\|\,.
\end{multline}
The integrability on $[0,\pi/2]$ of the terms in the second line follows from
the proof of the existence of $\w_+$, and of the first term in the first
line--from the proof of Thm.\ \ref{aspos}. Finally, the proof of
integrability of $\|\av(\si,.)\|_\infty\|\piti(\si)\psi\|$ is very similar to
the case of $\|\av(\si,.)\|^2_\infty$ (see relations \eqref{figae} and
\eqref{pibound}). This shows that the lhs of \eqref{ufu} has a limit for
$\tau\to\pi/2$. But $u_0(0,\tau)$ has a unitary strong limit operator
$u_0(0,\pi/2)$, so $f(\x)u(\tau,0)\psi$ converges strongly to some
vector~$\chi$. It follows that for each vector $\varphi$ there is
\begin{equation}
 (\varphi,\chi)=\lim_{\tau\to\pi/2}(\varphi,f(\x)u(\tau,0)\psi)
 =\lim_{\tau\to\pi/2}(u(0,\tau)f(\x)^*\varphi,\psi)=(\w_+f(\x)^*\varphi,\psi)\,,
\end{equation}
so $\chi=f(\x)\w_+^*\psi$ and the property \eqref{smoothfw} follows.
\end{proof}
It should be clear from the proof of the above theorem that the only
potential obstacle to the asymptotic completeness is the behavior of
$\pav\cdot\ev(\tau,\x)$ in the neighborhood of $\x^2=1$. If the norm
$\|\pav\cdot\ev(\tau,.)\|_{K,\infty}$ may be replaced by
$\|\pav\cdot\ev(\tau,.)\|_\infty$, then in the proof of convergence of
\eqref{ufu} the function $f$ may be replaced by $1$, and one obtains the
following.
\begin{thm}\label{ascompl}
If the assumptions of Theorem \ref{wave} hold with the condition
\eqref{wave2} replaced by
\begin{equation}\label{wave4}
 \int_{-\pi/2}^{\pi/2}\big(\tfrac{\pi}{2}-|\tau|\big)\|\pav\cdot\ev(\tau,.)\|_\infty\,d\tau<\infty\,,
\end{equation}
then the wave operators $\w_\pm$ are unitary.
\end{thm}

\section{Back to spacetime and conclusions}\label{scatspti}

Let us remind the reader, that the class of scattering electromagnetic fields
$F(x)$ was announced in the abstract, and then discussed in Appendix
\ref{aselm}. The bounds they satisfy were adapted to the field $\ev$ in
Section \ref{fieldsgauges}, where we also anticipated that the resulting
estimates satisfy the assumptions of Theorems~\ref{aspos} and \ref{wave}. We
now formulate the results of these theorems in natural spacetime terms, with
the use of relation \eqref{Uu} providing the link between the oscillator
picture evolution operator $u(\tau,\si)$ and the Schr\"odinger operator
$U(t,s)$.

The asymptotic variables $\xti_\pm$ provided by Theorem \ref{aspos} may be
obtained by
\begin{equation}\label{xS}
 \exp{\big[i\qv\cdot \xti_\pm]}
 =\slim_{t\to\pm\infty}U(0,t)\exp{\Big[i\qv\cdot\frac{\x}{\tm}\Big]}\,U(t,0)\,,
\end{equation}
where we used relation \eqref{NFN}. Now they have a natural interpretation of
asym\-ptotic velocity operators. The wave operators $\w_\pm$ of
Theorem~\ref{wave} are now
\begin{equation}\label{waveS}
  \w_\pm=\slim_{t\to\pm\infty}U(0,t)N(t)\,,
\end{equation}
and their conjugates are given by
\begin{equation}\label{waveSconj}
  \w^*_\pm=\slim_{t\to\pm\infty}N(t)^*U(t,0)E_\pm\,.
\end{equation}

It is visible that with our choice of gauge the definition of wave operators
in the present formalism does not need further corrections (e.g.\ of the
Dollard type) to compensate the long-range character of the interaction. The
joint spectrum of $\xti_+$ (and similarly $\xti_-$) covers the whole space
$\mR^3$ and outside $\xti_+^2=1$ ($\xti_-^2=1$) is absolutely continuous.
Whether this continuity extends to $\mR^3$, in which case $E_\pm=\mathbb{1}$
and the asymptotic completeness holds, is an open problem. The difficulty
lies in the null-asymptotic behavior of electromagnetic potentials (and
fields), which is slower than assumed in usual analyses of time-decaying
potentials. The problem may be due to the usual inconsistency in systems of
the considered type: the Lorentz symmetry of the electromagnetic field vs
Galilean symmetry of the~Schr\"odinger equation.

The existence of the wave operators in the given form, despite the problems
with asymptotic completeness, is a further argument for our main point:
appropriate choice of gauge eliminates at least some of the infrared
problems. The gauge condition found suitable in the present context is
characterized by Eq.\ \eqref{xA}. This property closely parallels the gauge
condition obtained in the analysis of the asymptotic behavior of the Dirac
field inside the lightcone \cite{her95}. In our opinion this may have
interesting implications for quantum electrodynamics as well, although in
that context the non-locality of the gauge will be a problem (cf.\
\eqref{figae}).

\section{Acknowledgements}

I am grateful to Jan Derezi\'nski for important literature hints and to
Pawe{\l} Duch for an interesting discussion.

\section*{Appendix}

\appendix

\section{Domain relations}\label{inequalities}

\begin{lem}
Let $\av,\,\pav\av \in L^\infty(\mR^3)$, so that the operator $\piv^2+\x^2$
is essentially self-adjoint on $C^\infty_0(\mR^3)$. Then the following
inequalities are satisfied
\begin{gather}
 \|\piv^2\varphi\|^2+\|\x^2\varphi\|^2\leq\|(\piv^2+\x^2)\varphi\|^2+6\|\varphi\|^2\,,\label{pix}\\[1ex]
 \|(\x\cdot\piv+\piv\cdot\x)\varphi\|^2
 \leq2(\|\piv^2\varphi\|^2+\|\x^2\varphi\|^2)+3\|\varphi\|^2\,.\label{xppx}
\end{gather}
Therefore,
\begin{equation}
 \Dc(\piv^2)\cap\Dc(\x^2)=\Dc(\piv^2+\x^2)\subseteq \Dc(\x\cdot\piv+\piv\cdot\x)\,.
\end{equation}
\end{lem}
\begin{proof}
For $\varphi\in C^\infty_0(\mR^3)$ the following identity is easily obtained:
\begin{equation}\label{pidel}
 \ov{\varphi}\,\piv^2\varphi+\varphi\,\ov{\piv^2\varphi}=-\Delta|\varphi|^2+2|\piv\varphi|^2\,.
\end{equation}
Multiplying this by $\x^2$ and integrating one obtains (integrate by parts
the first term on the rhs)
\[
 \Rp(\piv^2\varphi,\x^2\varphi)=-3\|\varphi\|^2+\||\x|\piv\varphi\|^2\geq -3\|\varphi\|^2\,,
\]
so \eqref{pix} for $\varphi\in C^\infty_0(\mR^3)$ easily follows. For
\eqref{xppx} note that $|\x\cdot\piv\varphi|^2\leq\x^2|\piv\varphi|^2$, so
using again \eqref{pidel} we have
\begin{multline}
 \|\x\cdot\piv\varphi\|^2\leq\tfrac{1}{2}\int\x^2\Delta|\varphi|^2+\Rp(\piv^2\varphi,\x^2\varphi)\\
 \leq3\|\varphi\|^2+\|\piv^2\varphi\|\|\x^2\varphi\|
 \leq \tfrac{1}{2}(\|\piv^2\varphi\|^2+\|\x^2\varphi\|^2)+3\|\varphi\|^2\,.
\end{multline}
Moreover, splitting
$\x\cdot\piv=\tfrac{1}{2}(\x\cdot\piv+\piv\cdot\x)+\tfrac{3}{2}i$ one has
\[
 \|\x\cdot\piv\varphi\|^2=\tfrac{1}{4}\|(\x\cdot\piv+\piv\cdot\x)\varphi\|^2
 +\tfrac{9}{4}\|\varphi\|^2\,.
\]
Substituting this on the lhs of the last inequality one obtains \eqref{xppx}
for\linebreak \mbox{$\varphi\in C^\infty_0(\mR^3)$}.

If $\varphi\in\Dc(\piv^2+\x^2)$, then it is approximated in the graph norm of
this operator by a sequence $\varphi_n\in C^\infty_0(\mR^3)$. The
inequalities \eqref{pix} and \eqref{xppx} show that then this is also a
Cauchy sequence in the graph norms of $\piv^2$, $\x^2$ and
$\x\cdot\piv+\piv\cdot\x$. These operators are closed, so $\varphi$ is in
their domains, and relations \eqref{pix} and \eqref{xppx} are satisfied on
$\Dc(\piv^2+\x^2)$. The inclusion
$\Dc(\piv^2)\cap\Dc(\x^2)\subseteq\Dc(\piv^2+\x^2)$ is obvious.
\end{proof}

\section{Asymptotic behavior of electromagnetic fields}\label{aselm}

We discuss here the decay properties of electromagnetic fields in the
Minkowski spacetime language. For more extensive discussion and explanation
of details we refer the reader to \cite{her17}.

Fields $\Ev$ and $\Bv$ are the $3$-vector parts of the Minkowski tensor
$F_{ab}(x)$ given by
\begin{equation}
 E^i(t,\x)=F^{i0}(x)\,,\qquad B^i(t,\x)=-\sum_{jk}\vep^{ijk}F^{jk}(x)\,;
\end{equation}
here and in the rest of these remarks: $x=(t,\x)$; $a$, $b$ are Minkowski
indices and $i$, $j$, $k$ are $3$-space indices. Also, we note that the field
$\Cv$ defined by \eqref{C} is the $3$-vector space part of the field
\begin{equation}\label{CF}
 C^a(x)=F^a{}_b(x)x^b\,.
\end{equation}

Electromagnetic fields present in scattering situations of the field theory
have the form
\begin{equation}
 F=F^\ret+F^\inc=F^\adv+F^\out\,,
\end{equation}
 where $F^\ret/F^\adv$ is the
retarded/advanced field of an electromagnetic current:
\begin{equation}
 F^{\mathrm{ret}/\mathrm{adv}}_{ab}(x)
 =4\pi\int D_{\mathrm{ret}/\mathrm{adv}}(x-y)[\p_aJ_b(y)-\p_bJ_a(y)]dy
\end{equation}
 and
$F^\inc/F^\out$ is a free incoming/outgoing field. In remote past and future
the current $J$ is assumed to tend to asymptotic currents, homogeneous of
degree~$-3$, with support inside the lightcone. Similarly, the free fields
are produced as radiation fields of currents with the asymptotic behavior of
the same type. We discuss the asymptotic behavior of fields in the half-space
$t\geq0$ with the use of representation $F=F^\adv+F^\out$; the case $t\leq0$
with $F=F^\ret+F^\inc$ is analogous.

For $t\geq0$ and $t+|\x|$ tending to infinity, the dominant contribution to
the advanced field comes from the asymptotic outgoing current. Therefore, in
this region this field is homogeneous of degree $-2$, hence the field
$\Cv^\adv$ is homogeneous of degree $-1$, and the field $\pav\cdot\Cv^\adv$
is homogeneous of degree $-2$. Thus for the advanced part the bounds
\eqref{EC} and  \eqref{divC} are satisfied. In fact, the $-2$ homogeneity of
$\pav\cdot\Cv^\adv$ implies that for the corresponding part $\ev^\adv$ of
$\ev$ there is $|\pav\cdot\ev^\adv(\tau,\x)|\leq\con/\cos\tau$, which is a
bound of the type~\eqref{eboundstr} sufficient for the validity of Thm.\
\ref{ascompl}.\footnote{The above reasoning is sufficient as it stands if the
asymptotic currents do not have oscillating contributions. However, a closer
analysis of the cases such as the classical Dirac field current, which does
have oscillating asymptotic terms, confirms the above bounds also in this
case. We do not enter here into more detailed discussion of this point and
refer the reader to \cite{her95} and \cite{her17} for a more explicit
discussion of the setting.}

For the free outgoing field and its Lorenz gauge potential we use the
integral representations valid for solutions of the wave equation (see e.g.\
\cite{her17}):
\begin{gather}\label{Fout}
 F^\out_{ab}(x)=-\frac{1}{2\pi}\int\Big[l_a\ddot{V}_b(x\cdot l,l)
 -l_b\ddot{V}_a(x\cdot l,l)\Big]\,d^2l\,,\\
 A^\out_a(x)=-\frac{1}{2\pi}\int\dot{V}_a(x\cdot l,l)\,d^2l\,.\label{Aout}
\end{gather}
Here $l$ represents a vector on the future lightcone and $d^2l$ is the
Lorentz-invariant measure on the set of future null directions, applicable to
functions of $l$ homogeneous of degree $-2$. The function $V$ describes the
future null asymptote of the potential and the field:
\begin{equation}
 \lim_{R\to\infty}RA^\out(x+Rl)=V(x\cdot l,l)\,,\quad
 \lim_{R\to\infty}RF^\out(x+Rl)=2l\wedge \dot{V}(x\cdot l,l)\,,
\end{equation}
and $\dot{V}(s,l)=\p V(s,l)/\p s$. Moreover, $V(s,l)$ is homogeneous of
degree $-1$, $l\cdot V(s,l)=0$, $V(s,l)\to0$ for $s\to\infty$, and
$\dot{V}(s,l)$ (together with its low orders derivatives in the cone
variable) is bounded by $\con(1+|s|)^{-1-\vep}$ for some $\vep>0$. Using this
bound it is easy to show that
\begin{multline}
 |A^\out(x)|\leq \con\int\frac{d\W(\lv)}{(1+|x^0-\x\cdot\lv|)^{1+\vep}}\\
 \leq \frac{\con}{1+|x^0|+|\x|}\Big[\frac{\theta(x^2)}{(1+|x^0|-|\x|)^\vep}
 +\theta(-x^2)\Big]\,,
\end{multline}
where $\lv$ represents unit $3$-vectors and $d\W(\lv)=\sin\vartheta
d\vartheta d\varphi$ is the solid angle measure, with $\vartheta$ measured
with respect to $\x$. In the region $x^2\geq0$ the field $F^\out$ is
estimated similarly as above (with the bound of $\ddot{V}(s,l)$ by
\mbox{$\con(1+|s|)^{-2-\vep}$} following from earlier assumptions), but in
the region $x^2<0$ one first has to integrate in \eqref{Fout} once by parts
with respect to $\vartheta$ angle variable. In this way one finds
\begin{equation}
 |F^\out(x)|\leq\frac{\con}{1+|x^0|+|\x|}\Big[\frac{1}{(1+||x^0|-|\x||)^{1+\vep}}
 +\frac{\theta(-x^2)}{1+|x^0|+|\x|}\Big]\,.
\end{equation}
The estimate of $\Ev^\out$ in \eqref{EC} is obviously satisfied. We turn to
the field $\Cv^\out$. The relation \eqref{CF} shows that this is again a
solution of the wave equation with the integral representation
\begin{equation}
 C^\out_a(x)=-\frac{1}{2\pi}\int \dot{W}_a(x\cdot l,l)\,d^2l\,.
\end{equation}
Denote $L_{ab}=l_a\p/\p l^b-l_b\p/\p l^a$--the intrinsic differential
operator on the lightcone. With the use of identity $\int L_{ab}F(l)d^2l=0$
valid for each $C^1$-function~$F$, homogeneous of degree $-2$, one shows that
\begin{equation}
 W^a(s,l)=V^a(s,l)-s\dot{V}^a(s,l)-L^{ab}V_b(s,l)\,.
\end{equation}
As \mbox{$l\cdot W(s,l)=0$}, one has $\p\cdot C(x)=0$, so
\begin{equation}
 \pav\cdot\Cv^\out(x)=\frac{1}{2\pi}\int l_0\ddot{W}^0(x\cdot l,l)\,d^2l\,.
\end{equation}
Now, it is easy to see that the function $\dot{W}(s,l)$ has the same decay
properties in $|s|$ as function $\dot{V}(s,l)$, so $\Cv^\out(x)$ and
$\pav\cdot\Cv^\out(x)$ are similarly estimated as $A^\out(x)$ and
$F^\out(x)$, respectively. This is sufficient to satisfy the bounds
\eqref{EC} and \eqref{divC}.

\section{Differentiation of operators $\tilde{k}$}\label{difk}

For a symmetric operator $\tilde{k}$ defined by \eqref{tik} suppose that
$\Dc(h_0)$ is contained in the domain of $\dot{k}$ and of the weak commutator
$[k(\tau),h(\tau)]_\mathrm{w}$ defined for $\varphi,\,\psi\in\Dc(h_0)$ by:
\begin{equation}
 (\varphi,[k(\tau),h(\tau)]_\mathrm{w}\psi)
 =(k(\tau)\varphi,h(\tau)\psi)-(h(\tau)\varphi,k(\tau)\psi)\,.
 \end{equation}
In this case we have
 $\p_\tau(\varphi,\tilde{k}\psi)=(\varphi,\big(\tilde{\dot{k}}-i\wt{[k,h]_\mathrm{w}}\big)\psi)$,
and upon integrating we can omit the product with $\varphi$ to obtain
\begin{equation}
 [\tilde{k}(\tau)-k(0)]\psi=\int_0^\tau\big(\tilde{\dot{k}}-i\wt{[k,h]_\mathrm{w}}\big)(\rho)\psi\,d\rho\,,
\end{equation}
from which the differential form follows
\begin{equation}
 \dot{\tilde{k}}\psi=\big(\tilde{\dot{k}}-i\wt{[k,h]_\mathrm{w}}\big)\psi\,.
\end{equation}
For our basic observables it is easy to show, that the weak commutators are
equal to those calculated naively, so in particular
\begin{equation}\label{xh}
 [\x,h]_\mathrm{w}=i\piv\,.
\end{equation}
Also, rather obviously, $[h,h]_\mathrm{w}=0$, so
$\dot{\tilde{h}}=\tilde{\dot{h}}$.

\end{document}